\documentclass[sn-mathphys-num]{sn-jnl}

\usepackage{mathtools,amssymb,bm}
\usepackage{microtype}
\newtheorem{theorem}{Theorem}[section]
\newtheorem{proposition}[theorem]{Proposition}
\newtheorem{lemma}[theorem]{Lemma}
\newtheorem{corollary}[theorem]{Corollary}
\theoremstyle{definition}
\newtheorem{assumption}[theorem]{Assumption}
\theoremstyle{remark}

\newcommand{\R}{\mathbb R}
\newcommand{\cH}{\mathcal H}
\newcommand{\cK}{\mathcal K}
\newcommand{\cS}{\mathcal S}
\newcommand{\cV}{\mathcal V}
\newcommand{\cW}{\mathcal W}
\newcommand{\ind}{\iota}
\newcommand{\prox}{\operatorname{prox}}
\newcommand{\ran}{\operatorname{ran}}
\newcommand{\ip}[2]{\left\langle #1,#2\right\rangle}
\newcommand{\norm}[1]{\left\lVert #1\right\rVert}
\newcommand{\normH}[1]{\left\lVert #1\right\rVert_H}
\newcommand{\avg}{\operatorname{avg}}
\newcommand{\exch}{\operatorname{exch}}

\begin{document}

\title[A Prediction--Correction Analysis of Two-Way Block Splitting in Distributed Learning]{A Prediction--Correction Analysis of Two-Way Block Splitting in Distributed Learning}

\author*[1,2]{\fnm{Xiaofei} \sur{Wu}}\email{xfwu1016@ynu.edu.cn}

\affil*[1]{\orgdiv{School of Mathematics and Statistics},
\orgname{Yunnan University},
\orgaddress{\city{Kunming}, \state{Yunnan}, \postcode{650500},
\country{China}}}

\affil[2]{\orgname{Yunnan Key Laboratory of Statistical Modeling and Data Analysis},
\orgaddress{Yunnan University, \city{Kunming}, \state{Yunnan},
\postcode{650500}, \country{China}}}

\abstract{
This note revisits the convergence of the Parikh--Boyd two-way block-splitting
algorithm for large-scale distributed learning through the He--Yuan
prediction--correction framework. Simultaneous row--column partitioning is also
relevant to hybrid federated learning, where data may be heterogeneous in both
samples and features. We lift the reduced iteration to an equal-dimensional
product space and reconstruct the primal and dual coordinates omitted by its
implementation. The induced orthogonal-complement structure establishes exact
iteration-by-iteration equivalence with the published updates. A mixed
variational-inequality representation then yields a fundamental descent
inequality, global convergence, an ergodic complexity bound, and
current-iterate residual estimates under standard convexity, solvability,
exact-subproblem, and invariant-initialization assumptions. The analysis also
shows that the reduced state recursion is a metric proximal point iteration.
No strong convexity, differentiability, or full-rank condition is imposed. The
derivation clarifies which algebraic initialization conditions allow the
reduced implementation to inherit the full-space convergence and complexity
guarantees without modifying its local updates.}


\maketitle

\section{Introduction}

\subsection{Distributed block-splitting model}

Large-scale learning may require a data matrix to be partitioned by both
observations and features. This creates column-consensus constraints between
feature copies and row-exchange constraints between partial fitted values.
Parikh and Boyd enforce them by local proximal and graph projections followed
by averaging and exchange projections \citep{ParikhBoyd2010,ParikhBoyd2014}.
The same two-way organization arises in hybrid federated learning when parties
differ in sample coverage and feature spaces; recent model-matching and
security studies illustrate this connection
\citep{ZhangEtAl2024Hybrid,YuEtAl2025Hybrid}. Their learning architectures need
not use the convex algorithm studied here, but they motivate simultaneous
row--column partitioning.

Let $A\in\R^{m\times n}$ be partitioned into $M\times N$ blocks
$A_{ij}\in\R^{m_i\times n_j}$. Parikh and Boyd consider
\citep{ParikhBoyd2010,ParikhBoyd2014}
\begin{equation}
 \min_{x,y}\quad l(y)+r(x)
 \qquad\text{subject to}\qquad y=Ax,
 \label{eq:original}
\end{equation}
where
\[
 l(y)=\sum_{i=1}^M l_i(y_i),
 \qquad
 r(x)=\sum_{j=1}^N r_j(x_j).
\]
The row and column partitions are conformable with $A$.
Here $y_i\in\R^{m_i}$ is the $i$th row block of the fitted value, $x_j\in
\R^{n_j}$ is the $j$th feature block, and $A_{ij}$ maps the latter into the
former. Introduce a local copy $x_{ij}$ of $x_j$ and a local contribution
$y_{ij}=A_{ij}x_{ij}$. With
\[
 C_{ij}:=\{(y_{ij},x_{ij}):y_{ij}=A_{ij}x_{ij}\}
\]
and the indicator $\ind_{C_{ij}}$, problem \eqref{eq:original} is equivalently
written as
\begin{equation}
 \begin{aligned}
 \min_{\{y_i,x_j,y_{ij},x_{ij}\}}\quad
 &\sum_{i=1}^M l_i(y_i)+\sum_{j=1}^N r_j(x_j)
  +\sum_{i=1}^M\sum_{j=1}^N\ind_{C_{ij}}(y_{ij},x_{ij})\\
 \text{subject to}\quad
 &x_{ij}=x_j\quad(i=1,\ldots,M;\ j=1,\ldots,N),\\
 &y_i=\sum_{j=1}^N y_{ij}\quad(i=1,\ldots,M).
 \end{aligned}
 \label{eq:PB-lifted}
\end{equation}
The first constraint is column consensus, while the second is the row
exchange constraint. The indicator incorporates each local linear relation
without communicating the block $A_{ij}$.
We use the convention
\[
 \prox_{\psi}(q):=\arg\min_z\{\psi(z)+\tfrac12\norm{z-q}^2\}.
\]

\subsection{The block-splitting iteration }

For a penalty parameter $\rho>0$, one iteration \citep{ParikhBoyd2010} consists of the following
eight updates:
\begin{align}
 y_i^{k+1/2}
 &=\prox_{l_i/\rho}(y_i^k-\widetilde y_i^k),
 &&\label{eq:PB3}\\
 x_j^{k+1/2}
 &=\prox_{r_j/\rho}(x_j^k-\widetilde x_j^k),
 &&\label{eq:PB4}\\
 (y_{ij}^{k+1/2},x_{ij}^{k+1/2})
 &=\Pi_{ij}(y_{ij}^k+\widetilde y_i^k,
                  x_j^k-\widetilde x_{ij}^k),
 &&\label{eq:PB5}\\
 x_j^{k+1}
 &=\avg(x_j^{k+1/2},\{x_{ij}^{k+1/2}\}_{i=1}^M),
 &&\label{eq:PB6}\\
 (y_i^{k+1},\{y_{ij}^{k+1}\}_{j=1}^N)
 &=\exch(y_i^{k+1/2},\{y_{ij}^{k+1/2}\}_{j=1}^N),
 &&\label{eq:PB7}\\
 \widetilde x_j^{k+1}
 &=\widetilde x_j^k+x_j^{k+1/2}-x_j^{k+1},
 &&\label{eq:PB8}\\
 \widetilde y_i^{k+1}
 &=\widetilde y_i^k+y_i^{k+1/2}-y_i^{k+1},
 &&\label{eq:PB9}\\
 \widetilde x_{ij}^{k+1}
 &=\widetilde x_{ij}^k+x_{ij}^{k+1/2}-x_j^{k+1}.
 &&\label{eq:PB10}
\end{align}
Here $k+1/2$ denotes a local prediction before consistency projection.
Equations (3) and (4) apply the proximal maps of the row loss $l_i$ and column
regularizer $r_j$, while (5) uses the Euclidean graph projection
$\Pi_{ij}=P_{C_{ij}}$, so that
$y_{ij}^{k+1/2}=A_{ij}x_{ij}^{k+1/2}$. The average in (6) assigns the mean of
the global prediction and its $M$ copies to every copy, and the exchange map
in (7) projects onto $y_i=\sum_jy_{ij}$; their formulas appear in
\eqref{eq:avg}--\eqref{eq:exch}. Finally, (8)--(10) update the retained scaled
dual coordinates $\widetilde x_j^k$, $\widetilde y_i^k$, and
$\widetilde x_{ij}^k$. The omitted coordinate
$\widetilde y_{ij}^k=-\widetilde y_i^k$ is reconstructed in
Section~\ref{sec:dual-structure}. Thus (3)--(5) are independent local
predictions, (6)--(7) enforce consistency, and (8)--(10) update the reduced
dual state.

\subsection{Scope, contributions, and assumptions}

This note does not modify (3)--(10). It restores equal-dimensional primal and
dual variables, characterizes the orthogonal complement of the
consensus--exchange subspace, and proves exact iteration-by-iteration recovery
of the reduced algorithm. The resulting mixed variational inequality provides
the prediction inequality, metric-compatible correction, Fej\'er decrease,
global convergence, an ergodic gap bound, and current-iterate residual rates
within the He--Yuan framework
\citep{HeYuan2012Contraction,HeYuan2023}. The distinction between squared
residual and residual-norm rates is consistent with related ADMM analyses
\citep{DengLaiPengYin2017,HeYuan2012Rate,HeYuan2015,ZamaniEtAl2024}.

Parikh and Boyd already identify graph-projection splitting as an
ADMM/Douglas--Rachford method, state convergence, and derive eliminated-dual
relations \citep{ParikhBoyd2014}. Our contribution is therefore not a first
convergence guarantee, but an explicit invariant-space derivation that verifies
the prediction--correction conditions and yields the stated rate bounds. For
background on distributed ADMM and splitting--contraction methods, see
\citep{BoydEtAl2011,HeTaoYuan2012}.

\section{Full product-space lifting and primal equivalence}
\label{sec:full-lifting}

\subsection{The lifted problem}

Retain the graph subspaces $C_{ij}$ from Section~1 and write
$\Pi_{ij}:=P_{C_{ij}}$. Introduce $x_{ij}\in\R^{n_j}$ and
$y_{ij}\in\R^{m_i}$ and work in the Euclidean product space
\begin{equation}
 \cH:=
 \left(\prod_{i=1}^M\R^{m_i}\right)
 \times\left(\prod_{j=1}^N\R^{n_j}\right)
 \times\left(\prod_{i=1}^M\prod_{j=1}^N\R^{m_i}\right)
 \times\left(\prod_{i=1}^M\prod_{j=1}^N\R^{n_j}\right),
 \label{eq:H}
\end{equation}
with its standard inner product, where
$z=(\{y_i\},\{x_j\},\{y_{ij}\},\{x_{ij}\})\in\cH$. Set
\begin{equation}
 \Phi(z):=
 \sum_{i=1}^M l_i(y_i)+\sum_{j=1}^N r_j(x_j)
 +\sum_{i=1}^M\sum_{j=1}^N\ind_{C_{ij}}(y_{ij},x_{ij}).
 \label{eq:Phi}
\end{equation}
The remaining consistency constraints form the closed linear subspace
\begin{equation}
 \cS:=\left\{z\in\cH:
 x_{ij}=x_j\ \ (i,j),\quad
 y_i=\sum_{j=1}^N y_{ij}\ \ (i=1,\ldots,M)
 \right\}.
 \label{eq:S}
\end{equation}
Thus the lifted formulation is
\begin{equation}
 \min_{z\in\cS}\Phi(z).
 \label{eq:lifted}
\end{equation}

\begin{proposition}[Equivalence of the primal models]\label{prop:primal-equivalence}
Problem \eqref{eq:lifted} is equivalent to \eqref{eq:original}. More
precisely, every feasible point of \eqref{eq:lifted} with finite objective satisfies
$y_i=\sum_j A_{ij}x_j$, and every feasible point of
\eqref{eq:original} with finite objective has the lift $x_{ij}=x_j$ and
$y_{ij}=A_{ij}x_j$, with the same objective value.
\end{proposition}

\begin{proof}
The indicators and \eqref{eq:S} give
$y_{ij}=A_{ij}x_{ij}$, $x_{ij}=x_j$, and $y_i=\sum_jy_{ij}$, hence
$y_i=\sum_jA_{ij}x_j$. The stated lift proves the converse and preserves the
objective.
\end{proof}

For the analysis, duplicate the entire product-space variable:
\begin{equation}
 \min_{\xi,\eta\in\cH}\quad \Phi(\xi)+\ind_{\cS}(\eta)
 \qquad\text{subject to}\qquad \xi-\eta=0.
 \label{eq:two-copy}
\end{equation}
The two copies and the dual variable contain all four coordinate families in
\eqref{eq:H}; the implementation eliminates invariant coordinates.

With $u^k$ denoting the scaled multiplier, two-block ADMM for
\eqref{eq:two-copy} reads
\begin{subequations}\label{eq:full-admm}
\begin{align}
 \xi^{k+1}
 &=\prox_{\Phi/\rho}(\eta^k-u^k),
 \label{eq:full-prox}\\
 \eta^{k+1}
 &=P_{\cS}(\xi^{k+1}+u^k),
 \label{eq:full-proj}\\
 u^{k+1}
 &=u^k+\xi^{k+1}-\eta^{k+1}.
 \label{eq:full-dual}
\end{align}
\end{subequations}

The analytical variable $\xi^{k+1}$ is the full-space collection of the
reduced half-step quantities in (3)--(5), not an additional update; the
integer index follows standard ADMM notation, while $k+1/2$ distinguishes the
local predictions from post-projection variables
\citep{ParikhBoyd2010,ParikhBoyd2014}.

\begin{assumption}\label{ass:main}
Each $l_i$ and $r_j$ is proper, lower semicontinuous, and convex. The lifted
problem \eqref{eq:lifted} admits a KKT point: there exist
$z^\star\in\cS$ and $\lambda^\star\in\cS^\perp$ such that
\begin{equation}
 -\lambda^\star\in\partial\Phi(z^\star).
 \label{eq:kkt-assumption}
\end{equation}
The penalty parameter satisfies $\rho>0$, all proximal and projection
subproblems are solved exactly, and the full-space iteration is initialized
with $\eta^0\in\cS$ and $u^0\in\cS^\perp$.
\end{assumption}

This standard saddle-point assumption requires no strong convexity,
differentiability, or full-rank condition \citep{BoydEtAl2011}. Zero
initialization satisfies the invariant-space requirement. With arbitrary
primal starts, the analysis may begin after the first projection; the initial
column-dual zero-sum condition remains necessary for the simplified first
average to equal the full projection.

\section{Orthogonal-complement structure and exact reduction}
\label{sec:dual-structure}

\subsection{The orthogonal complement}

Let
\[
 \cK:=
 \left(\prod_{i=1}^M\prod_{j=1}^N\R^{n_j}\right)
 \times\left(\prod_{i=1}^M\R^{m_i}\right)
\]
with the standard product inner product. Define the constraint-residual
operator $R$ and display its adjoint as follows:
\begin{equation}
 \begin{aligned}
 R:\cH&\longrightarrow\cK,\\
 Rz&:=\left(
 \{x_{ij}-x_j\}_{i,j},
 \left\{y_i-\sum_{j=1}^N y_{ij}\right\}_{i}
 \right),\\[1mm]
 R^*(\alpha,\gamma)
 &:=\left(
 \{\gamma_i\}_{i},
 \left\{-\sum_{i=1}^M\alpha_{ij}\right\}_{j},
 \{-\gamma_i\}_{i,j},
 \{\alpha_{ij}\}_{i,j}
 \right).
 \end{aligned}
 \label{eq:R}
\end{equation}
The components of $R^*(\alpha,\gamma)$ follow directly by expanding
$\ip{Rz}{(\alpha,\gamma)}$ in the coordinate order of \eqref{eq:H}. Since
$Rz=0$ is exactly \eqref{eq:S}, $\cS=\ker R$ and, in finite dimensions,
\[
 \cS^\perp=(\ker R)^\perp=\ran R^*.
\]
Here $\cS^\perp$ denotes the orthogonal complement.

For $q\in\cH$, $[q]_{x_{ij}}$ denotes its component in the slot labelled by
$x_{ij}$; the notation for the other labelled slots is analogous.

\begin{lemma}[Coordinate description of $\cS^\perp$]\label{lem:Sperp}
A vector $q\in\cH$ belongs to $\cS^\perp$ if and only if there exist
vectors $\alpha_{ij}\in\R^{n_j}$ and $\gamma_i\in\R^{m_i}$ such that
\begin{equation}
 [q]_{x_{ij}}=\alpha_{ij},\qquad
 [q]_{x_j}=-\sum_{i=1}^M\alpha_{ij},
 \qquad
 [q]_{y_i}=\gamma_i,\qquad
 [q]_{y_{ij}}=-\gamma_i.
 \label{eq:Sperp-coordinates}
\end{equation}
Equivalently,
\begin{equation}
 [q]_{y_{ij}}=-[q]_{y_i},
 \qquad
 [q]_{x_j}+\sum_{i=1}^M [q]_{x_{ij}}=0.
 \label{eq:Sperp-identities}
\end{equation}
\end{lemma}

\begin{proof}
The formula for $R^*$ gives \eqref{eq:Sperp-coordinates}, and
$(\ker R)^\perp=\ran R^*$ makes the characterization necessary and
sufficient.
\end{proof}

\begin{lemma}[Projection and dual invariance]\label{lem:invariance}
Suppose $\eta^0\in\cS$ and $u^0\in\cS^\perp$; in particular, the zero
initialization has this property. Then, for every $k\geq0$,
\begin{equation}
 \eta^k\in\cS,\qquad u^k\in\cS^\perp,
 \label{eq:invariance}
\end{equation}
and \eqref{eq:full-admm} reduces to
\begin{subequations}\label{eq:orth-admm}
\begin{align}
 \xi^{k+1}&=\prox_{\Phi/\rho}(\eta^k-u^k),\label{eq:orth-prox}\\
 \eta^{k+1}&=P_{\cS}\xi^{k+1},\label{eq:orth-proj}\\
 u^{k+1}&=u^k+P_{\cS^\perp}\xi^{k+1}.\label{eq:orth-dual}
\end{align}
\end{subequations}
Moreover, even without $u^k\in\cS^\perp$, one projection--dual update puts
$u^{k+1}$ in $\cS^\perp$.
\end{lemma}

\begin{proof}
Equation \eqref{eq:full-proj} gives $\eta^{k+1}\in\cS$. Also,
\[
 u^{k+1}
 =\xi^{k+1}+u^k-P_{\cS}(\xi^{k+1}+u^k)
 =P_{\cS^\perp}(\xi^{k+1}+u^k),
\]
which belongs to $\cS^\perp$ for arbitrary $u^k$. Under the induction
hypothesis $P_{\cS}u^k=0$, so
$P_{\cS}(\xi^{k+1}+u^k)=P_{\cS}\xi^{k+1}$. Substitution into the dual
update gives \eqref{eq:orth-dual}.
\end{proof}

\subsection{Iteration-by-iteration equivalence with (3)--(10)}

For fixed $j$ and $i$, respectively, define
\[
 \cS_{x,j}:=\{(x_j,\{x_{ij}\}_{i=1}^M):x_{ij}=x_j\ \forall i\},
 \qquad
 \cS_{y,i}:=\{(y_i,\{y_{ij}\}_{j=1}^N):y_i=\textstyle\sum_jy_{ij}\}.
\]
Then $\cS$ is the Cartesian product of these row and column subspaces,
which act on disjoint coordinate families. Their orthogonal projections are
explicit. For fixed $j$, let
\begin{equation}
 \bar c:=\avg(c,\{c_i\}_{i=1}^M)
 :=\frac{c+\sum_{i=1}^M c_i}{M+1},\qquad
 P_{\cS_{x,j}}(c,c_1,\ldots,c_M)
 =(\bar c,\bar c,\ldots,\bar c),
 \label{eq:avg}
\end{equation}
For fixed $i$, put $d=(c-\sum_{j=1}^N c_j)/(N+1)$. Then
\begin{equation}
 \exch(c,\{c_j\}_{j=1}^N)
 :=(c-d,c_1+d,\ldots,c_N+d)
 =P_{\cS_{y,i}}(c,c_1,\ldots,c_N).
 \label{eq:exch}
\end{equation}
These are precisely the published average and exchange maps
\citep{ParikhBoyd2010,ParikhBoyd2014}.

Following the original notation \citep{ParikhBoyd2010}, the graph pair is $(y_{ij},x_{ij})$ and the
objective blocks are $l_i,r_j$.  Our $\prox_{h/\rho}$ is its parameter-$\rho$ proximal operator.
The retained coordinates of the full variables match the published iterates
as follows \citep{ParikhBoyd2010}:
\[
 \begin{aligned}
 [\eta^k]_{y_i}&=y_i^k,
 & [\eta^k]_{x_j}&=x_j^k,\\
 [\eta^k]_{y_{ij}}&=y_{ij}^k,
 & [\eta^k]_{x_{ij}}&=x_j^k,\\
 [\xi^{k+1}]_{y_i}&=y_i^{k+1/2},
 & [\xi^{k+1}]_{x_j}&=x_j^{k+1/2},\\
 [\xi^{k+1}]_{y_{ij}}&=y_{ij}^{k+1/2},
 & [\xi^{k+1}]_{x_{ij}}&=x_{ij}^{k+1/2}.
 \end{aligned}
\]
In particular, $[\eta^k]_{x_{ij}}=x_j^k$ is an eliminated consistency
copy, whereas $x_{ij}^{k+1/2}$ is a retained local prediction.

\begin{proposition}[Exact recovery of the reduced algorithm]\label{prop:exact-recovery}
Assume that $\eta^k\in\cS$ and $u^k\in\cS^\perp$, and define the reduced
dual coordinates by
\begin{equation}
 \widetilde y_i^k:=[u^k]_{y_i},\qquad
 \widetilde x_j^k:=[u^k]_{x_j},\qquad
 \widetilde x_{ij}^k:=[u^k]_{x_{ij}}.
 \label{eq:reduced-duals}
\end{equation}
Then \eqref{eq:orth-admm}, coordinate by coordinate, is exactly the
published iteration \eqref{eq:PB3}--\eqref{eq:PB10} in
\citep{ParikhBoyd2010}.

Conversely, any reduced sequence satisfying
\eqref{eq:PB3}--\eqref{eq:PB10} and initialized with
\begin{equation}
 y_i^0=\sum_{j=1}^N y_{ij}^0,
 \qquad
 \widetilde x_j^0+\sum_{i=1}^M\widetilde x_{ij}^0=0
 \label{eq:reduced-init}
\end{equation}
has a unique full-space reconstruction after setting
$[\eta^k]_{x_{ij}}=x_j^k$ and
$[u^k]_{y_{ij}}=-\widetilde y_i^k$.
\end{proposition}

\begin{proof}
Separability of \eqref{eq:Phi} makes \eqref{eq:orth-prox} split into the two
proximal mappings and the $MN$ graph projections. Lemma
\ref{lem:Sperp} gives
\[
 [u^k]_{y_{ij}}=-\widetilde y_i^k,\qquad
 [u^k]_{x_j}+\sum_i [u^k]_{x_{ij}}=0.
\]
Since $\eta^k\in\cS$, its local coordinate
$[\eta^k]_{x_{ij}}$ equals $x_j^k$. Therefore
the four inputs to the local proximal step are exactly those in
\eqref{eq:PB3}--\eqref{eq:PB5}; the plus sign in the $y_{ij}$ input is the
identity $-[u^k]_{y_{ij}}=\widetilde y_i^k$. Equations
\eqref{eq:avg}--\eqref{eq:exch} give \eqref{eq:PB6}--\eqref{eq:PB7}, and
the remaining displayed equations are the corresponding coordinates of
\eqref{eq:orth-dual}.

Conversely, averaging implies
\[
 \widetilde x_j^{k+1}+\sum_i\widetilde x_{ij}^{k+1}
 =\widetilde x_j^k+\sum_i\widetilde x_{ij}^k,
\]
so the second relation in \eqref{eq:reduced-init} is invariant. The exchange formula similarly makes
the omitted $y_{ij}$-dual update the negative of \eqref{eq:PB9}. Lemma
\ref{lem:Sperp} then reconstructs the omitted coordinates, and the full
iteration follows. Thus the equivalence holds at every iteration, not only
at a limit point.
\end{proof}

The indices $k+1/2$ and $k+1$ are the pre- and post-projection stages,
corresponding to $\xi^{k+1}$ and $\eta^{k+1}$. Zero initialization satisfies
\eqref{eq:reduced-init}, and the identity
$[u^k]_{y_{ij}}=-\widetilde y_i^k$ reconstructs an omitted coordinate rather
than imposing another assumption.

\section{Mixed VI and prediction--correction representation}

Set $\lambda^k:=\rho u^k$ and introduce the invariant state space
\begin{equation}
 \cV:=\cS\times\cS^\perp,
 \qquad v=(\eta,\lambda).
 \label{eq:V}
\end{equation}
To retain the local proximal point $\xi$ while using the consistency
coordinate $\eta=P_{\cS}\xi$, define
\begin{equation}
 \cW:=\{w=(\xi,\eta,\lambda):
 \eta=P_{\cS}\xi,\ \lambda\in\cS^\perp\}
 \label{eq:W}
\end{equation}
and the linear operator
\begin{equation}
 F(w):=(\lambda,0,\eta-\xi).
 \label{eq:F}
\end{equation}
For $w,w'\in\cW$, one has
\begin{equation}
 \ip{w-w'}{F(w)-F(w')}
 =\ip{\lambda-\lambda'}{\eta-\eta'}=0,
 \label{eq:restricted-skew}
\end{equation}
because $\eta-\eta'\in\cS$ and
$\lambda-\lambda'\in\cS^\perp$. Thus $F$ is skew, and in particular
monotone, on $\cW$.

Consider the mixed VI
\begin{equation}
 w^\star\in\cW,\qquad
 \Phi(\xi)-\Phi(\xi^\star)
 +\ip{w-w^\star}{F(w^\star)}\geq0
 \quad\forall w\in\cW.
 \label{eq:mvi}
\end{equation}

\begin{proposition}[Equivalence of the mixed VI and the KKT system]
\label{prop:vi-kkt}
A point $w^\star=(\xi^\star,\eta^\star,\lambda^\star)$ solves
\eqref{eq:mvi} if and only if
\begin{equation}
 \xi^\star=\eta^\star\in\cS,\qquad
 \lambda^\star\in\cS^\perp,\qquad
 -\lambda^\star\in\partial\Phi(\xi^\star).
 \label{eq:kkt}
\end{equation}
Consequently, $\xi^\star$ solves \eqref{eq:lifted}.
\end{proposition}

\begin{proof}
In \eqref{eq:mvi}, vary $\lambda$ freely in the linear space
$\cS^\perp$. Since $\eta^\star-\xi^\star=-P_{\cS^\perp}\xi^\star$ also
belongs to $\cS^\perp$, the resulting inequality for all $\lambda$ forces
$\eta^\star-\xi^\star=0$. Every $\xi\in\cH$ occurs in a point of $\cW$ by
choosing $\eta=P_{\cS}\xi$. The remaining inequality therefore says
\[
 \Phi(\xi)-\Phi(\xi^\star)
 +\ip{\xi-\xi^\star}{\lambda^\star}\geq0
 \quad\forall\xi\in\cH,
\]
which is equivalent to $-\lambda^\star\in\partial\Phi(\xi^\star)$.
The converse is immediate from the subgradient inequality and
orthogonality.
\end{proof}

Given $v^k=(\eta^k,\lambda^k)\in\cV$, define the predictor
\begin{subequations}\label{eq:predictor}
\begin{align}
 \widehat\xi^k
 &:=\prox_{\Phi/\rho}(\eta^k-\rho^{-1}\lambda^k),
 \label{eq:pred-xi}\\
 \widehat\eta^k
 &:=P_{\cS}\widehat\xi^k,
 \label{eq:pred-eta}\\
 \widehat\lambda^k
 &:=\lambda^k+\rho(\widehat\xi^k-\widehat\eta^k).
 \label{eq:pred-lambda}
\end{align}
\end{subequations}
Thus $\widehat w^k=(\widehat\xi^k,\widehat\eta^k,
\widehat\lambda^k)\in\cW$ and
$\widehat v^k=(\widehat\eta^k,\widehat\lambda^k)\in\cV$.
By Lemma \ref{lem:invariance} and $\lambda^k=\rho u^k$, these symbols are
identified with the full ADMM stages by
\[
 \widehat\xi^k=\xi^{k+1},\qquad
 \widehat\eta^k=\eta^{k+1},\qquad
 \widehat\lambda^k=\lambda^{k+1}.
\]
Thus the hat marks the predictor and agrees with the integer-indexed
full-space ADMM block. Its coordinates are precisely the half-step local
primal quantities in the reduced notation \citep{ParikhBoyd2010}.

\begin{theorem}[Prediction inequality and framework matrices]
\label{thm:prediction}
For every $w=(\xi,\eta,\lambda)\in\cW$, the predictor
\eqref{eq:predictor} satisfies
\begin{equation}
 \Phi(\xi)-\Phi(\widehat\xi^k)
 +\ip{w-\widehat w^k}{F(\widehat w^k)}
 \geq
 \ip{v-\widehat v^k}{Q(v^k-\widehat v^k)},
 \label{eq:prediction-inequality}
\end{equation}
where, on $\cV=\cS\times\cS^\perp$,
\begin{equation}
 Q:=\begin{pmatrix}
 \rho I_{\cS}&0\\[1mm]
 0&\rho^{-1}I_{\cS^\perp}
 \end{pmatrix}.
 \label{eq:Q}
\end{equation}
The original block-splitting iteration is the correction
\begin{equation}
 v^{k+1}=v^k-M(v^k-\widehat v^k),
 \qquad M:=I_{\cV},
 \label{eq:correction}
\end{equation}
and one may take
\begin{equation}
 H:=Q,\qquad
 G:=Q+Q^*-M^*HM=Q.
 \label{eq:HMG}
\end{equation}
Consequently,
\begin{equation}
 HM=Q\quad\text{and}\quad G\succ0,
 \label{eq:C2C3}
\end{equation}
which are precisely the metric compatibility condition and positive
definiteness of the so-called profit matrix in the prediction--correction framework
\citep{HeYuan2023}.
\end{theorem}

\begin{proof}
The optimality condition of \eqref{eq:pred-xi} is
\[
 0\in\partial\Phi(\widehat\xi^k)+\lambda^k
       +\rho(\widehat\xi^k-\eta^k).
\]
Using \eqref{eq:pred-lambda} and setting
$d_\eta^k:=\eta^k-\widehat\eta^k$, this becomes
\begin{equation}
 -\widehat\lambda^k+\rho d_\eta^k
 \in\partial\Phi(\widehat\xi^k).
 \label{eq:pred-subgrad}
\end{equation}
Hence
\begin{equation}
 \Phi(\xi)-\Phi(\widehat\xi^k)
 +\ip{\xi-\widehat\xi^k}{\widehat\lambda^k}
 \geq \rho\ip{\xi-\widehat\xi^k}{d_\eta^k}.
 \label{eq:subgrad-ineq}
\end{equation}
Because $d_\eta^k\in\cS$ and
$\eta=P_{\cS}\xi$, the right-hand side equals
$\rho\ip{\eta-\widehat\eta^k}{d_\eta^k}$. Moreover,
\eqref{eq:pred-lambda} gives
\begin{equation}
 \widehat\eta^k-\widehat\xi^k
 =\rho^{-1}(\lambda^k-\widehat\lambda^k).
 \label{eq:pred-feas}
\end{equation}
Adding
$\ip{\lambda-\widehat\lambda^k}{\widehat\eta^k-\widehat\xi^k}$
to \eqref{eq:subgrad-ineq} yields exactly
\eqref{eq:prediction-inequality} with \eqref{eq:Q}. Lemma
\ref{lem:invariance} shows that \eqref{eq:predictor} produces the actual next
state, so \eqref{eq:correction} holds with $M=I_{\cV}$. Finally, $Q$ is
self-adjoint and positive definite on $\cV$, and the identities
\eqref{eq:HMG}--\eqref{eq:C2C3} follow directly.
\end{proof}

Positive definiteness in \eqref{eq:C2C3} holds on the invariant state space
$\cS\times\cS^\perp$ after redundant coordinates are removed; it is neither
an additional objective regularity assumption nor an added proximal term.

\section{Fundamental inequality and global convergence}

Equip $\cV$ with
\begin{equation}
 \normH{v}^2:=\ip{v}{Hv}
 =\rho\norm{\eta}^2+\rho^{-1}\norm{\lambda}^2.
 \label{eq:Hnorm}
\end{equation}

\begin{theorem}[Fundamental inequality]\label{thm:fundamental}
For every $w\in\cW$ and every $k\geq0$,
\begin{align}
 &\Phi(\xi)-\Phi(\widehat\xi^k)
 +\ip{w-\widehat w^k}{F(\widehat w^k)}
 \notag\\
 &\qquad\geq\frac12\left(
 \normH{v-v^{k+1}}^2-\normH{v-v^k}^2
 +\normH{v^k-v^{k+1}}^2\right).
 \label{eq:fundamental}
\end{align}
\end{theorem}

\begin{proof}
By Theorem \ref{thm:prediction}, $Q=H$ and
$\widehat v^k=v^{k+1}$. Apply the polarization identity
\[
 2\ip{v-v^{k+1}}{H(v^k-v^{k+1})}
 =\normH{v-v^{k+1}}^2-\normH{v-v^k}^2
  +\normH{v^k-v^{k+1}}^2
\]
to the right-hand side of \eqref{eq:prediction-inequality}.
\end{proof}

Let $\cV^\star$ denote the $v=(\eta,\lambda)$ components of the solutions of
\eqref{eq:mvi}.

\begin{corollary}[Fej\'er decrease]\label{cor:fejer}
For every $v^\star\in\cV^\star$,
\begin{equation}
 \normH{v^{k+1}-v^\star}^2
 \leq \normH{v^k-v^\star}^2
      -\normH{v^k-v^{k+1}}^2.
 \label{eq:fejer}
\end{equation}
In particular,
\begin{equation}
 \sum_{k=0}^{\infty}\normH{v^k-v^{k+1}}^2
 \leq \normH{v^0-v^\star}^2<\infty.
 \label{eq:summable}
\end{equation}
\end{corollary}

\begin{proof}
Set $w=w^\star$ in \eqref{eq:fundamental}. By the VI at $w^\star$ and the
restricted skew identity \eqref{eq:restricted-skew},
\[
 \Phi(\xi^\star)-\Phi(\widehat\xi^k)
 +\ip{w^\star-\widehat w^k}{F(\widehat w^k)}\leq0.
\]
Inequality \eqref{eq:fejer} follows. Summing it proves
\eqref{eq:summable}.
\end{proof}

\begin{theorem}[Global convergence]\label{thm:global}
Under Assumption \ref{ass:main}, the sequence $v^k$ converges to a point
$v^\infty=(z^\infty,\lambda^\infty)\in\cV^\star$. Moreover,
\begin{equation}
 \eta^k\to z^\infty,\qquad
 \xi^{k+1}\to z^\infty,\qquad
 \lambda^k\to\lambda^\infty,
 \label{eq:global-limits}
\end{equation}
where $z^\infty$ solves \eqref{eq:lifted}. Through Proposition
\ref{prop:exact-recovery}, all retained coordinates of (3)--(10) converge to
the corresponding coordinates of this primal-dual solution.
\end{theorem}

\begin{proof}
Fej\'er monotonicity makes $\{v^k\}$ bounded, and
\eqref{eq:summable} gives $v^{k+1}-v^k\to0$. Equation
\eqref{eq:pred-feas} then gives
$\widehat\xi^k-\widehat\eta^k\to0$, so every cluster point has equal primal
coordinates. Along a convergent subsequence, the subgradient in
\eqref{eq:pred-subgrad} converges to $-\bar\lambda$ at $\bar z$. Closedness of
the subdifferential graph yields \eqref{eq:kkt}. A Fej\'er sequence with a
cluster point in its target set converges to that point, proving
\eqref{eq:global-limits}.
\end{proof}

\section{Ergodic and pointwise rates}

\subsection{Ergodic \texorpdfstring{$O(1/K)$}{O(1/K)} rate}

For $K\geq1$, define the averages of the prediction points
\begin{equation}
 \bar w^K:=\frac1K\sum_{k=0}^{K-1}\widehat w^k
 =(\bar\xi^K,\bar\eta^K,\bar\lambda^K).
 \label{eq:ergodic-average}
\end{equation}
Since $\cW$ is convex, $\bar w^K\in\cW$.

\begin{theorem}[Ergodic mixed-VI bound]\label{thm:ergodic}
For every comparison point $w=(\xi,\eta,\lambda)\in\cW$,
\begin{equation}
 \Phi(\bar\xi^K)-\Phi(\xi)
 +\ip{\bar w^K-w}{F(w)}
 \leq \frac{\normH{v-v^0}^2}{2K}.
 \label{eq:ergodic-gap}
\end{equation}
Furthermore,
\begin{equation}
 \bar\xi^K-\bar\eta^K
 =\frac{\lambda^K-\lambda^0}{\rho K},
 \label{eq:ergodic-feasibility}
\end{equation}
and hence the averaged feasibility residual is $O(1/K)$. If
$w^\star$ is any solution of \eqref{eq:mvi}, then
\begin{equation}
 \left|\Phi(\bar\xi^K)-\Phi(\xi^\star)\right|
 \leq \frac{\normH{v^\star-v^0}^2}{2K}
 +\norm{\lambda^\star}\,\norm{\bar\xi^K-\bar\eta^K}
 =O(K^{-1}).
 \label{eq:ergodic-objective}
\end{equation}
\end{theorem}

\begin{proof}
Use \eqref{eq:restricted-skew} to reverse the sign and arguments of the
left-hand side of \eqref{eq:fundamental}. Discard its nonnegative step term, sum from
$k=0$ to $K-1$, and telescope. Convexity of $\Phi$ and linearity of the
remaining term give \eqref{eq:ergodic-gap}. Summing
\eqref{eq:pred-lambda} gives \eqref{eq:ergodic-feasibility}; boundedness of
$\{\lambda^k\}$ gives the rate. At a solution,
$\xi^\star=\eta^\star$ and $F(w^\star)=(\lambda^\star,0,0)$. Applying
\eqref{eq:ergodic-gap} with $w=w^\star$ yields the upper objective bound,
while $-\lambda^\star\in\partial\Phi(\xi^\star)$ yields the matching lower
bound. Since $\lambda^\star\in\cS^\perp$,
$\ip{\bar\xi^K-\xi^\star}{\lambda^\star}
=\ip{\bar\xi^K-\bar\eta^K}{\lambda^\star}$, which proves
\eqref{eq:ergodic-objective}.
\end{proof}

\subsection{Nonergodic residual rate}

Because ADMM rates depend on the selected performance measure
\citep{DengLaiPengYin2017,ZamaniEtAl2024}, define the squared pointwise residual by
\begin{equation}
 \mathcal R_{k+1}^2
 :=\normH{v^k-v^{k+1}}^2
 =\rho\norm{\eta^k-\eta^{k+1}}^2
  +\rho^{-1}\norm{\lambda^k-\lambda^{k+1}}^2.
 \label{eq:pointwise-residual}
\end{equation}
This quantity is an explicit KKT residual certificate. Indeed, from
\eqref{eq:pred-subgrad},
\begin{equation}
 a^{k+1}:=-\lambda^{k+1}
            +\rho(\eta^k-\eta^{k+1})
 \in\partial\Phi(\xi^{k+1}),
 \label{eq:kkt-subgradient}
\end{equation}
while
\begin{equation}
 \xi^{k+1}-\eta^{k+1}
 =\rho^{-1}(\lambda^{k+1}-\lambda^k).
 \label{eq:point-feasibility}
\end{equation}
Consequently,
\begin{equation}
\mathcal R_{k+1}^2
 =\rho^{-1}\norm{a^{k+1}+\lambda^{k+1}}^2
  +\rho\norm{\xi^{k+1}-\eta^{k+1}}^2.
 \label{eq:kkt-residual-identity}
\end{equation}
The two terms in \eqref{eq:kkt-residual-identity} measure stationarity and
copy feasibility; $\eta^{k+1}\in\cS$ and
$\lambda^{k+1}\in\cS^\perp$ hold exactly.

\begin{lemma}[Firm nonexpansiveness of the state map]\label{lem:firm}
Let $T:\cV\to\cV$ map any state $v$ to the predictor state $\widehat v$
generated by \eqref{eq:predictor}. Then
\begin{equation}
 \normH{T v-Tv'}^2
 \leq\ip{Tv-Tv'}{H(v-v')}
 \qquad\forall v,v'\in\cV.
 \label{eq:firm}
\end{equation}
In particular, $T$ is nonexpansive and
$\mathcal R_{k+2}\leq\mathcal R_{k+1}$.
\end{lemma}

\begin{proof}
Write the prediction inequality \eqref{eq:prediction-inequality} once for
the input $v$ and test it at the predictor generated by $v'$, and once with
$v,v'$ interchanged. Add the two inequalities. The convex terms cancel and,
by \eqref{eq:restricted-skew}, the sum of the $F$ terms is zero. Since
$Q=H$ and $M=I$, the remaining terms give
\[
 0\geq\normH{Tv-Tv'}^2-\ip{Tv-Tv'}{H(v-v')},
\]
which is \eqref{eq:firm}. Cauchy--Schwarz gives nonexpansiveness; applying it
to $v^k,v^{k+1}$ gives
$\mathcal R_{k+2}\leq\mathcal R_{k+1}$.
\end{proof}

Because $HM=Q=H\succ0$ implies $M=I$, the correction accepts the predictor.
Firm nonexpansiveness gives, for the maximal monotone relation
$\mathcal A:=H(T^{-1}-I)$, the resolvent identity
$T=(I+H^{-1}\mathcal A)^{-1}$. Hence the reduced state recursion is an
$H$-metric proximal point algorithm, consistent with proximal interpretations
of distributed splitting, ADMM, and the Chen--Teboulle algorithm
\citep{EcksteinBertsekas1992,ParikhBoyd2014Proximal,CaiGuHeYuan2013,
Becker2019ChenTeboulle}. This does not replace its separated primal updates by
a monolithic subproblem.

\begin{theorem}[Current-iterate residual bounds]\label{thm:pointwise}
For every $v^\star\in\cV^\star$ and $K\geq1$,
\begin{equation}
 \mathcal R_K^2
 \leq \frac{\normH{v^0-v^\star}^2}{K},
 \qquad
 \mathcal R_K
 \leq \frac{\normH{v^0-v^\star}}{\sqrt K}.
 \label{eq:pointwise-bigO}
\end{equation}
Moreover,
\begin{equation}
 \mathcal R_K^2=o(K^{-1}),
 \qquad
 \mathcal R_K=o(K^{-1/2}).
 \label{eq:pointwise-littleo}
\end{equation}
Thus the squared stationarity-plus-feasibility certificate in
\eqref{eq:kkt-residual-identity} has a nonergodic $O(1/K)$ rate, while the
residual norm has a nonergodic $O(K^{-1/2})$ rate.
\end{theorem}

\begin{proof}
Corollary \ref{cor:fejer} gives
$\sum_{k=0}^{K-1}\mathcal R_{k+1}^2
\leq\normH{v^0-v^\star}^2$. Lemma \ref{lem:firm} makes the summands
nonincreasing, so
\[
 K\mathcal R_K^2
 \leq\sum_{k=0}^{K-1}\mathcal R_{k+1}^2
 \leq\normH{v^0-v^\star}^2,
\]
which proves \eqref{eq:pointwise-bigO}. A nonnegative, nonincreasing,
summable sequence $a_k$ satisfies $ka_k\to0$: apply monotonicity on the
index interval $\{\lfloor k/2\rfloor,\ldots,k\}$ and use vanishing of the
tail sum. Taking $a_k=\mathcal R_k^2$ proves
\eqref{eq:pointwise-littleo}.
\end{proof}

Residual monotonicity upgrades the usual best-iterate estimate from summability
to the current-iterate little-$o$ bound. The norm rate still takes a square
root; a faster norm rate requires additional regularity.

\section{Conclusion}

The full-space lifting reconstructs the eliminated dual coordinates of block
splitting and identifies its invariant state space. On this space the original
iteration has $Q=H=G$ and $M=I$, is an $H$-metric proximal point method, and
admits Fej\'er decrease, global convergence, and ergodic and nonergodic rate
bounds. Zero initialization satisfies the required invariant condition; without
stronger regularity, the squared current residual is $O(1/K)$ and its norm is
$O(K^{-1/2})$.

\section*{Declarations}

\subsection*{Funding}
Xiaofei Wu acknowledges support from the Visiting Scholar Program of the
Chern Institute of Mathematics and from the Talent Research Start-up Fund of
the School of Mathematics and Statistics, Yunnan University (grant no.
CZ305010126XX).

\subsection*{Competing interests}
The author declares no competing interests.

\subsection*{Data availability}
No datasets were generated or analyzed during the current study.

\bibliography{block_splitting_prediction_correction_note}

\end{document}